\pdfoutput=1
\documentclass[11pt,journal,onecolumn]{IEEEtran}

\usepackage[T1]{fontenc}
\usepackage[utf8]{inputenc}
\usepackage{amsmath,amsthm,amssymb,mathtools}
\usepackage{cite}
\usepackage{microtype}
\usepackage{booktabs,array}
\usepackage{needspace}
\usepackage{graphicx}
\usepackage{tikz}
\usepackage{pgfplots}
\usepackage{algorithm}
\usepackage{algpseudocode}
\usepackage{url}
\usepackage{xcolor}
\usepackage{hyperref}
\pgfplotsset{compat=1.18}

\hypersetup{
  pdftitle={Pair-Partition Constructions for CPM-Based Quantum LDPC Codes},
  pdfauthor={Koki Okada and Kenta Kasai},
  unicode=true,
  bookmarksnumbered=true,
  colorlinks=true,
  linkcolor=black,
  citecolor=black,
  urlcolor=black
}
\newcommand{\F}{\mathbb{F}}
\newcommand{\Z}{\mathbb{Z}}
\newcommand{\rank}{\operatorname{rank}}
\newcommand{\rowspan}{\operatorname{rowspan}}
\newcommand{\wt}{\operatorname{wt}}
\newcommand{\qcode}[3]{[\![#1,#2,#3]\!]}
\newcommand{\swp}[2]{(#1\,#2)}
\theoremstyle{definition}
\newtheorem{theorem}{Theorem}
\newtheorem{proposition}{Proposition}
\newtheorem{lemma}{Lemma}
\newtheorem{corollary}{Corollary}
\newtheorem{example}{Example}
\newtheorem{definition}{Definition}
\theoremstyle{remark}
\newtheorem{remark}{Remark}

\title{Pair-Partition Constructions for CPM-Based Quantum LDPC Codes}
\author{Koki~Okada and Kenta~Kasai%
\thanks{The authors are with the Institute of Science Tokyo, Tokyo, Japan
(e-mail: okada.k.3154@m.isct.ac.jp; kenta@ict.eng.isct.ac.jp).}}

\begin{document}
\maketitle

\begin{abstract}
We introduce the pair-partition (PP) construction of binary
Calderbank--Shor--Steane quantum low-density parity-check codes from
circulant permutation matrices.  A square array of pair partitions
imposes linear paired-difference equations on the CPM exponents and thereby
guarantees CSS orthogonality.  Pairing graphs derived from this array allow
the combinatorial design to be screened before exponent search.  We further
give and prove a complete algorithm for verifying lower bounds on
the quantum minimum distance of a fixed CSS lift.  The algorithm searches
for zero-syndrome vectors outside the opposing stabilizer row space, uses
only rigorously valid pruning rules, and finds no vector through a prescribed
weight if and only if the corresponding distance exceeds that weight.  For
fixed column weight, row weight, and search limit, cyclic symmetry makes the
combinatorial support-enumeration bound independent of the lift size,
although matrix preprocessing and row-space tests can still depend on
the lift size.  Thus the construction stage and the distance-verification stage are
both specified by directly checkable finite procedures.  Under checkable
rank and size conditions, a permanent construction also guarantees a logical
operator and a factorial distance upper bound determined only by the column
weight.
\end{abstract}

\section{Introduction}

Classical low-density parity-check (LDPC) codes were introduced as sparse
parity-check codes \cite{gallager1962}, and Tanner graphs became the standard
way to describe their local constraints and short cycles \cite{tanner1981}.
The Calderbank--Shor--Steane (CSS) construction turns two binary codes
satisfying an orthogonality condition into a quantum stabilizer code
\cite{calderbank1996,steane1996}.  Quantum LDPC codes combine these ideas by
requiring sparse stabilizer checks for syndrome extraction and sparse-graph
decoding \cite{mackay2004,breuckmann2021}.

Quasi-cyclic LDPC matrices built from circulant permutation matrices (CPMs)
describe a large sparse matrix by a small exponent array.  Their short cycles
can be checked by congruences on that array \cite{fossorier2004}.  We use the
more specific term CPM-LDPC for the cyclic part of the construction.  In the
CSS setting, two exponent arrays must also be coupled so that
\(H_XH_Z^{\mathsf T}=0\).  Algebraic quantum CPM-LDPC constructions were
developed in \cite{hagiwara2007}; related finite-length CPM and
affine-permutation constructions have pursued higher girth, finite rate, and
controlled distance \cite{komoto2025explicit,komoto2025bound,
kasai2025AffinePermutation,kasai2026OrthogonalityBarrier,
okadaKasai2026AffineCoset,okadaKasai2026SquareBase,
okadaKasai2026TwoBranch,okadaKasai2026RateTwoThirds}.

Product-type constructions establish positive-rate quantum LDPC families
with growing or almost-linear distance \cite{tillich2014,panteleev2022}.
Here we study the complementary finite-length problem of constructing fixed
CPM-LDPC matrices and verifying their structural properties and exact
distances with reproducible records.

We call the following procedure the \emph{pair-partition (PP) construction}.
Assign a pair partition of block-column indices to every pair of \(X\)- and
\(Z\)-block rows, equate the mixed CPM differences at the two endpoints of
each pair, and form the two CPM block matrices from exponent arrays satisfying
those linear equations.  The paired terms cancel over \(\mathbb F_2\), so
CSS orthogonality becomes a homogeneous linear condition.  The exponent
arrays are then screened for short cycles and low-weight logical operators.

The construction begins with a combinatorial design rule for the
pair-partition array \(M\), before any exponent is chosen.  Row and column
unions of the matchings in
\(M\) define pairing graphs.  A graph cycle of length \(r\) forces a Tanner
cycle of length at most \(2r\).  Thus a girth-eight search may immediately
discard any \(M\) whose pairing graphs have parallel edges or triangles; in
particular, girth eight requires \(L\ge2J\).  The complete protograph also
gives the universal Tanner-girth upper bound twelve.

After fixing an \(M\) that passes the pairing-graph design tests, we solve the paired-difference system, remove
exponent assignments that create short cycles or reproduce stored low-weight
logical patterns, and check the surviving lifted matrices directly.  The
stored patterns accelerate candidate generation.  Reported distance lower
bounds come from complete searches on the fixed matrices, and matching upper
bounds come from explicit non-stabilizer zero-syndrome vectors.  Cyclic
symmetry reduces the possible roots from \(LP\) to
\(L\); for fixed \(J,L,W\), the resulting support-search tree has a bound
independent of \(P\).  This bound covers combinatorial enumeration; total
runtime also includes construction of the length-\(LP\) matrices, rank
preprocessing, and stabilizer row-space tests whose cost may grow with \(P\).
Under the full-rank and size conditions satisfied by the reported codes, a
permanent construction additionally guarantees logical operators of weight
at most \((J+1)!\), giving a lift-size-independent upper bound on both CSS
distances.

The cancellation condition for CSS orthogonality does not require cyclic
translations: pairing columns with equal relative affine maps extends the
construction to affine permutation matrices (APMs).  The resulting APM-PP
presentation gives a directly checkable CSS certificate.  A structured
specialization identifies the APM-PP check matrices with the selected check
matrices, called active check matrices in the Kasai construction, while a
separate condition records the complementary parent rows.  The larger design space may also support
searches for codes suited to the GALA framework and AOD-based neutral-atom
implementations, with compatibility evaluated from each completed code's
rank, distance, short cycles, and implementation schedule.

The resulting list contains thirty-six exact-distance regular CSS codes.
Twenty-nine have girth six, while \(\qcode{472}{122}{16}\), \(\qcode{584}{150}{18}\),
\(\qcode{1112}{282}{20}\), \(\qcode{1336}{338}{22}\),
\(\qcode{1630}{656}{20}\), \(\qcode{1784}{450}{24}\), and
\(\qcode{2230}{896}{24}\) have girth eight.  Exact distances reach \(24\),
and rates reach \(0.574\), attained by the high-rate
\(\qcode{1414}{812}{12}\) code.
The \(\qcode{492}{170}{20}\) and \(\qcode{516}{178}{20}\) codes are compact
column-weight-four examples, and \(\qcode{708}{242}{22}\) extends this
sequence to distance 22.
Every parameter claim is tied to a fixed matrix and a reproducible record.

Section~2 defines the CPM-based PP construction and its cycle conditions.
Section~3 develops the pair-partition design rule.  Section~4 presents the
complete candidate-search procedure, including distance-guided forbidden
patterns.  Section~5 proves the distance lower-bound algorithm and describes
upper-bound witnesses.  Section~6 gives the numerical results.  Section~7
then identifies the part of the PP argument that extends to APMs, explains
the Kasai specialization, and gives a brief GALA/AOD design outlook.
Section~8 concludes the paper.

\section{CPM-Based CSS Construction}

\subsection{CPM block matrices}

All vectors are column vectors, and \(\rowspan(H)\) denotes the column-vector span of the transposed rows of \(H\).  Fix positive integers \(J\) and \(L\) with \(J<L\) and \(L\) even, and fix a prime lift size \(P\).  The rows of each parity-check matrix are grouped into \(J\) circulant block rows, and the \(L\) block columns are indexed by \(\Z_L\).  For \(s\in\F_P\), let \(C(s)\) denote the \(P\times P\) CPM whose row lift index \(a\) is incident to column lift index \(a-s\).  Define the exponent arrays \(E:=(e_{j\ell})\in\F_P^{J\times L}\) and \(D:=(d_{j\ell})\in\F_P^{J\times L}\), and set
\begin{equation}
H_X=\bigl(C(e_{j\ell})\bigr)_{j,\ell},\qquad
H_Z=\bigl(C(d_{j\ell})\bigr)_{j,\ell}.
\label{eq:blockmatrices}
\end{equation}
The matrices \(H_X\) and \(H_Z\) are \(JP\times LP\) binary block matrices built from CPM blocks, while \(e_{j\ell}\) and \(d_{j\ell}\) are the scalar exponents used in the algebraic constraints below.
Equivalently, both sides are CPM lifts of the complete all-one \(J\times L\) protograph.  Hence each matrix has \(JP\) rows, \(LP\) columns, row weight \(L\), and column weight \(J\).  Once CSS orthogonality holds, the number of encoded qubits is \(k=LP-\rank H_X-\rank H_Z\).
Ignoring rank dependencies gives the design-rate benchmark \(1-2J/L\); the
actual rate \(k/(LP)\) is computed from the two binary ranks.

\subsection{Pair partitions and mixed differences}

The PP construction uses a \(J\times J\) array \((M_{ij})_{0\le i,j<J}\) of pair partitions.  The array \((M_{ij})\) is chosen before the exponent labels.  Each \(M_{ij}\) is a pair partition of \(\Z_L\), namely a partition of the \(L\) block-column indices into \(L/2\) unordered pairs; hence \(L\) is even.  Any such array gives CSS orthogonality when the paired-difference equations below are satisfied.  Section~\ref{sec:pair-partition-design} then adds the conditions on \(M\) used to avoid forced Tanner cycles.

\begin{example}[\((3,8)\)-regular illustrative assignment]
\label{ex:pair-partition-array}
For a \((3,8)\)-regular illustrative assignment with \(P=53\), one
pair-partition array is shown in
Fig.~\ref{fig:pair-partition-array}.  For example,
\(\swp{0}{3}\swp{1}{2}\swp{4}{6}\swp{5}{7}\) denotes the four unordered
pairs \((0,3),(1,2),(4,6),(5,7)\).  Each entry is a pair partition of
\(\Z_8\).
\end{example}

\begin{figure}[H]
\centering
\footnotesize
\setlength{\arraycolsep}{3pt}
\[
(M_{ij})=
\left(
\begin{array}{c|c|c}
\swp{0}{3}\swp{1}{2}\swp{4}{6}\swp{5}{7} &
\swp{0}{2}\swp{1}{6}\swp{3}{7}\swp{4}{5} &
\swp{0}{6}\swp{1}{5}\swp{2}{3}\swp{4}{7} \\ \hline
\swp{0}{5}\swp{1}{3}\swp{2}{4}\swp{6}{7} &
\swp{0}{3}\swp{1}{5}\swp{2}{7}\swp{4}{6} &
\swp{0}{7}\swp{1}{4}\swp{2}{6}\swp{3}{5} \\ \hline
\swp{0}{1}\swp{2}{6}\swp{3}{5}\swp{4}{7} &
\swp{0}{5}\swp{1}{7}\swp{2}{4}\swp{3}{6} &
\swp{0}{2}\swp{1}{6}\swp{3}{7}\swp{4}{5}
\end{array}
\right).
\]
\caption{A \(3\times3\) pair-partition array on
\(\mathbb Z_8\) used in Example~\ref{ex:pair-partition-array}.}
\label{fig:pair-partition-array}
\end{figure}

For every pair \((u,v)\in M_{ij}\), impose the following linear equation over \(\F_P\):
\begin{equation}
 d_{j,u}-e_{i,u}=d_{j,v}-e_{i,v}.
 \label{eq:matchingconstraint}
\end{equation}
These equations force the mixed differences for the two endpoints of every pair to be equal.  To write them as one linear system, order the
\(2JL\) exponent variables as
\[
 x=\bigl(\operatorname{vec}(E),\operatorname{vec}(D)\bigr)^{\mathsf T}
 \in\F_P^{2JL}.
\]
Let \(A_M\in\F_P^{(J^2L/2)\times 2JL}\) have one row for every
\(\{u,v\}\in M_{ij}\), with coefficients \(-1,+1,+1,-1\) in the
columns of \(e_{i,u},e_{i,v},d_{j,u},d_{j,v}\), respectively, and zeros
elsewhere.  Then all paired-difference equations are exactly
\begin{equation}
 A_Mx=0.
 \label{eq:AM-system}
\end{equation}

\begin{example}[The full coefficient matrix \(A_M\)]
For the array in Fig.~\ref{fig:pair-partition-array}, order the equation rows
first by \((i,j)\) and then by the displayed pairs in \(M_{ij}\).  With the
column order \((\operatorname{vec}(E),\operatorname{vec}(D))\), the resulting
\(36\times48\) matrix is shown in Fig.~\ref{fig:AM-grid}.  Every row has
exactly four nonzero entries.  For \(\{u,v\}\in M_{ij}\), its colours mark
\(-1,+1,+1,-1\) in the columns of
\(e_{i,u},e_{i,v},d_{j,u},d_{j,v}\), respectively.  Row reduction of this
full matrix over \(\F_{53}\) produces the exponent solution space used in
the candidate search.
\end{example}

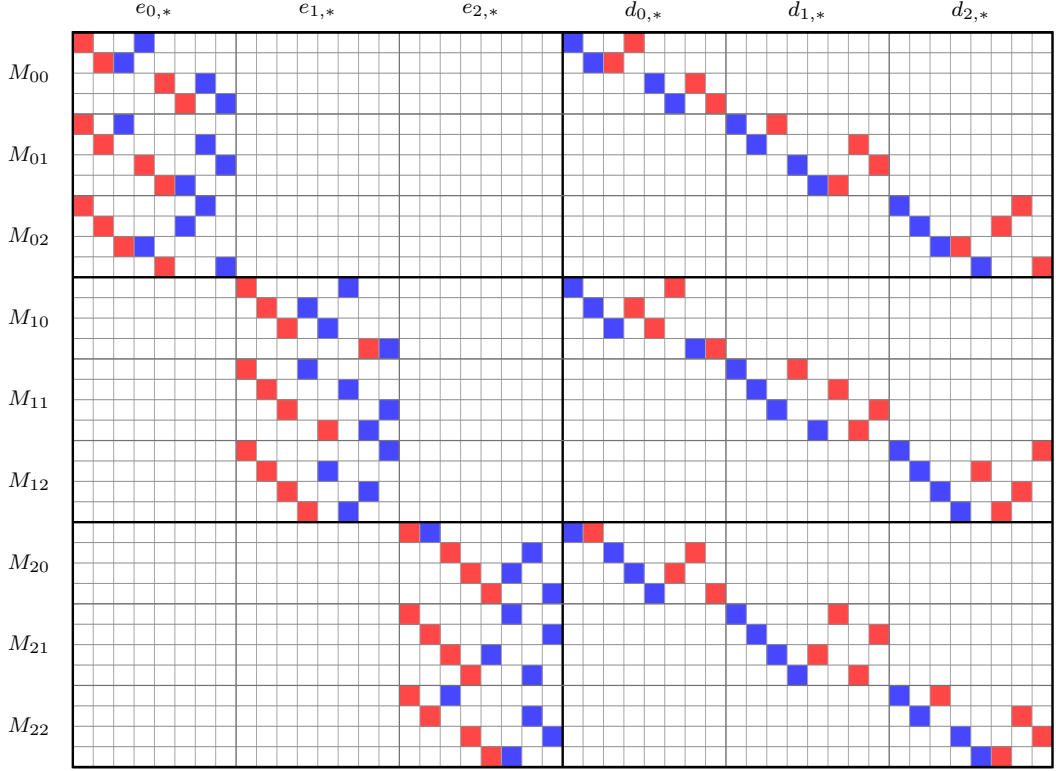
\begin{figure}[H]
\centering
\begin{tikzpicture}[x=0.27cm,y=-0.27cm,font=\scriptsize]
\fill[white] (0,0) rectangle (48,36);
\foreach[count=\r from 0] \i/\j/\u/\v in {
0/0/0/3,0/0/1/2,0/0/4/6,0/0/5/7,
0/1/0/2,0/1/1/6,0/1/3/7,0/1/4/5,
0/2/0/6,0/2/1/5,0/2/2/3,0/2/4/7,
1/0/0/5,1/0/1/3,1/0/2/4,1/0/6/7,
1/1/0/3,1/1/1/5,1/1/2/7,1/1/4/6,
1/2/0/7,1/2/1/4,1/2/2/6,1/2/3/5,
2/0/0/1,2/0/2/6,2/0/3/5,2/0/4/7,
2/1/0/5,2/1/1/7,2/1/2/4,2/1/3/6,
2/2/0/2,2/2/1/6,2/2/3/7,2/2/4/5}{
  \pgfmathtruncatemacro{\eu}{8*\i+\u}
  \pgfmathtruncatemacro{\ev}{8*\i+\v}
  \pgfmathtruncatemacro{\du}{24+8*\j+\u}
  \pgfmathtruncatemacro{\dv}{24+8*\j+\v}
  \fill[red!72]  (\eu,\r) rectangle ++(1,1);
  \fill[blue!72] (\ev,\r) rectangle ++(1,1);
  \fill[blue!72] (\du,\r) rectangle ++(1,1);
  \fill[red!72]  (\dv,\r) rectangle ++(1,1);
}
\foreach \x in {0,...,48}
  \draw[black!35,line width=0.04pt] (\x,0)--(\x,36);
\foreach \y in {0,...,36}
  \draw[black!45,line width=0.08pt] (0,\y)--(48,\y);
\foreach \x in {0,8,16,24,32,40,48}
  \draw[black!55,line width=0.45pt] (\x,0)--(\x,36);
\draw[black,line width=0.9pt] (24,0)--(24,36);
\foreach \y in {0,4,8,12,16,20,24,28,32,36}
  \draw[black!55,line width=0.45pt] (0,\y)--(48,\y);
\foreach \y in {0,12,24,36}
  \draw[black,line width=0.9pt] (0,\y)--(48,\y);
\draw[black,line width=0.9pt] (0,0) rectangle (48,36);

\node at (4,-1.0)  {$e_{0,*}$};
\node at (12,-1.0) {$e_{1,*}$};
\node at (20,-1.0) {$e_{2,*}$};
\node at (28,-1.0) {$d_{0,*}$};
\node at (36,-1.0) {$d_{1,*}$};
\node at (44,-1.0) {$d_{2,*}$};
\node[anchor=east] at (-0.6,2)  {$M_{00}$};
\node[anchor=east] at (-0.6,6)  {$M_{01}$};
\node[anchor=east] at (-0.6,10) {$M_{02}$};
\node[anchor=east] at (-0.6,14) {$M_{10}$};
\node[anchor=east] at (-0.6,18) {$M_{11}$};
\node[anchor=east] at (-0.6,22) {$M_{12}$};
\node[anchor=east] at (-0.6,26) {$M_{20}$};
\node[anchor=east] at (-0.6,30) {$M_{21}$};
\node[anchor=east] at (-0.6,34) {$M_{22}$};

\end{tikzpicture}
\caption{The full \(36\times48\) coefficient matrix \(A_M\) obtained from
Fig.~\ref{fig:pair-partition-array}.  White, blue, and red represent
\(0,+1\), and \(-1\), respectively.  Thick lines separate the exponent
blocks and the nine cells \(M_{ij}\).}
\label{fig:AM-grid}
\end{figure}

Thus the paired-difference part consists of \(J^2L/2\) raw linear equations,
although linear dependencies may remain.

\begin{example}[Number of variables and equations for \(J=3,L=8,P=53\)]
For the \((3,8)\)-regular illustrative assignment in Example~1, the exponent labels have \(2JL=48\) scalar variables.  Substituting the displayed \((M_{ij})\) into \eqref{eq:matchingconstraint} gives \(J^2L/2=36\) raw paired-difference equations over \(\F_{53}\).
\end{example}

\subsection{CSS orthogonality}

For CSS orthogonality, the pair-partition array controls common nonzero positions between rows of \(H_X\) and rows of \(H_Z\).  The paired-difference constraints are imposed so that these common positions occur with even multiplicity.

\begin{lemma}[CSS orthogonality]
\label{lem:cssorthogonality}
The pairing constraints \eqref{eq:matchingconstraint} imply \(H_XH_Z^{\mathsf T}=0\).
\end{lemma}

\begin{proof}
Fix a block row \(i\) of \(H_X\) and a block row \(j\) of \(H_Z\).  The corresponding \(P\times P\) block of \(H_XH_Z^{\mathsf T}\) is
\[
 \sum_{\ell=0}^{L-1} C(e_{i\ell})C(d_{j\ell})^{\mathsf T}
 =\sum_{\ell=0}^{L-1} C(e_{i\ell}-d_{j\ell}),
\]
where the sums are over \(\F_2\).  If \((u,v)\in M_{ij}\), then \eqref{eq:matchingconstraint} gives
\[
 e_{iu}-d_{ju}=e_{iv}-d_{jv}.
\]
The two CPMs contributed by \(u\) and \(v\) are therefore identical and cancel over \(\F_2\).  Since \(M_{ij}\) partitions all \(L\) block-column indices into pairs, the whole block sum is zero.  This holds for every \(i,j\), so \(H_XH_Z^{\mathsf T}=0\).
\end{proof}

\begin{theorem}[Pair-partition CPM-CSS construction]
\label{thm:construction}
Let \(L\) be even.  Suppose that a pair-partition array \((M_{ij})\) and exponent arrays \((e_{j\ell}),(d_{j\ell})\in\F_P^{J\times L}\) satisfy the paired-difference equations \eqref{eq:matchingconstraint}.  Then the lifted matrices \(H_X,H_Z\) in \eqref{eq:blockmatrices} satisfy \(H_XH_Z^{\mathsf T}=0\); every \(X\)-check and \(Z\)-check intersect in an even number of positions, which may exceed two.  The resulting code has block length \(n=LP\), row weight \(L\), and column weight \(J\).  Its dimension is
\[
 k=LP-\rank H_X-\rank H_Z.
\]
\end{theorem}

The theorem separates code validity from cycle screening.  Pairing makes
mixed \(X\)--\(Z\) overlaps cancel, while the CPM lift gives the claimed
degrees and length.  Cycle conditions are imposed separately below.

\begin{proof}
CSS orthogonality follows from Lemma~\ref{lem:cssorthogonality}.  Because every block of \(H_X\) and \(H_Z\) is a CPM, each lifted row has one nonzero entry in each of the \(L\) block columns, and each lifted column has one nonzero entry in each of the \(J\) block rows.  Hence the row and column weights are \(L\) and \(J\), respectively, and the block length is \(LP\).  The dimension formula is the CSS rank formula for \(H_XH_Z^{\mathsf T}=0\).
\end{proof}

The paired system has \(2JL\) scalar variables and \(J^2L/2\) raw equations.
This equation count is independent of \(P\); when field operations have unit
cost, its bit cost contains \(\log P\).  Explicitly forming the two sparse lifts, however, writes
\(2JLP\) nonzero entries, and binary rank computation acts on matrices whose
dimensions grow with \(P\).

\section{Pair-Partition Design}
\label{sec:pair-partition-design}

For fixed \(j\), let \(\Gamma_j^X\) be the edge-coloured multigraph on
\(\mathbb Z_L\) whose colour-\(i\) edges are the pairs in \(M_{ij}\).
Similarly, for fixed \(i\), let \(\Gamma_i^Z\) have colour-\(j\) edges
\(M_{ij}\).  Each colour class is a perfect matching, and two parallel edges
are counted as a cycle of length two.  These graphs depend only on \(M\), not
on the lift size or the exponents.

For a binary matrix \(H\), let \(g_{\mathrm T}(H)\) denote the girth of its
Tanner graph.  For a multigraph \(G\), let \(g(G)\) denote the length of its
shortest graph cycle, with \(g(G)=\infty\) when \(G\) has no cycle.

\Needspace{10\baselineskip}
\begin{proposition}[Necessary condition for avoiding forced 4-cycles]
\label{prop:M-no-reuse}
If both \(H_X\) and \(H_Z\) have Tanner girth at least six, then
\[
 M_{ij}\cap M_{ij'}=\varnothing\quad(j\ne j'),
 \qquad
 M_{ij}\cap M_{i'j}=\varnothing\quad(i\ne i').
\]
\end{proposition}

\begin{proof}
If the same pair occurs in \(M_{ij}\) and \(M_{ij'}\), then
\(\Gamma_i^Z\) has two parallel edges, which force a Tanner 4-cycle in
\(H_Z\).  If it occurs in \(M_{ij}\) and \(M_{i'j}\), then
\(\Gamma_j^X\) has two parallel edges and forces a Tanner 4-cycle in
\(H_X\).  These are the length-two cases of
Proposition~\ref{prop:pair-cycle-obstruction}.
\end{proof}

The PP definition guarantees CSS orthogonality without the no-reuse
condition.  The design stage adds this condition when the target Tanner
girth is at least six.

\subsection{Generating the pair-partition array}

The array \(M\) can be generated by a direct backtracking search on perfect
matchings of \(\mathbb Z_L\).  Let \(\mathcal P_L\) be the set of all
\((L-1)!!\) pair partitions of \(\mathbb Z_L\).  While filling the cells in
a fixed order, maintain the edge unions
\[
 R_i=\bigcup_{j:\,M_{ij}\text{ assigned}}M_{ij},\qquad
 C_j=\bigcup_{i:\,M_{ij}\text{ assigned}}M_{ij}.
\]
At cell \((i,j)\), try a matching \(Q\in\mathcal P_L\) only if
\begin{equation}
 Q\cap R_i=\varnothing,\qquad Q\cap C_j=\varnothing.
 \label{eq:M-generation-disjointness}
\end{equation}
These two tests enforce Proposition~\ref{prop:M-no-reuse} as the array is built.  For a target
Tanner girth \(g\in\{6,8\}\), also reject \(Q\) whenever either tentative
union \(R_i\cup Q\) or \(C_j\cup Q\) contains a graph cycle of length less
than \(g/2\).  Thus the girth-eight search rejects parallel edges and
triangles immediately.  If no matching survives, undo the preceding cell
and try its next matching; after all \(J^2\) cells have been assigned, output
the resulting \(M\) and form the linear system \(A_Mx=0\).

Enumerating every matching makes this procedure exhaustive; sampling or
randomly ordering \(\mathcal P_L\) gives the version used for a bounded
search.  Vertex relabelling permits the symmetry-breaking choice
\(M_{00}=\{\{0,1\},\{2,3\},\ldots,\{L-2,L-1\}\}\) without losing any
isomorphism class.  An edge-overlap or triangle test uses
\(O(JL)\) adjacency work per attempted cell, although the worst-case number
of backtracking nodes is exponential in \(J^2\).  This entire stage depends
only on \(J,L\), not on \(P,E\), or \(D\).

The six pairing graphs obtained from the array in
Example~\ref{ex:pair-partition-array} are shown in
Fig.~\ref{fig:example1-pairing-graphs}.
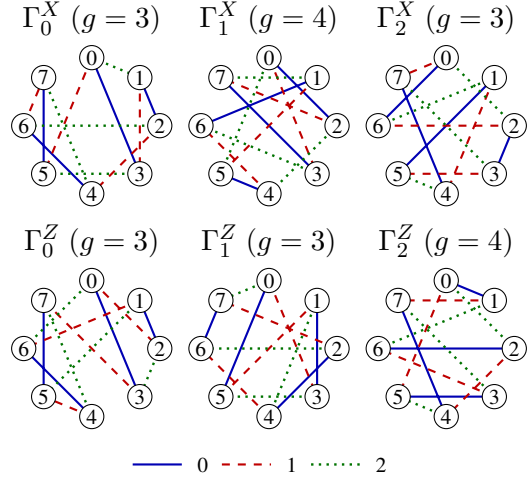
\begin{figure}[!t]
\centering
\begingroup
\tikzset{
  pairzero/.style={draw=blue!70!black,line width=0.8pt},
  pairone/.style={draw=red!75!black,line width=0.8pt,dashed},
  pairtwo/.style={draw=green!45!black,line width=0.9pt,dotted},
  pairvertex/.style={circle,draw=black,fill=white,inner sep=0pt,
    minimum size=3.2mm,font=\scriptsize}
}
\newcommand{\paircoords}{%
  \foreach \v/\a in {0/90,1/45,2/0,3/-45,4/-90,5/-135,6/180,7/135}
    \coordinate (v\v) at (\a:0.90cm);%
}
\newcommand{\pairnodes}{%
  \foreach \v in {0,...,7} \node[pairvertex] at (v\v) {\v};%
}
\newcommand{\pairtitle}[2]{\multicolumn{1}{c}{$#1\;(g=#2)$}}
\setlength{\tabcolsep}{3pt}
\begin{tabular}{ccc}
\pairtitle{\Gamma_0^X}{3} & \pairtitle{\Gamma_1^X}{4} &
\pairtitle{\Gamma_2^X}{3} \\[1mm]
\begin{tikzpicture}
  \paircoords
  \draw[pairzero] (v0)--(v3) (v1)--(v2) (v4)--(v6) (v5)--(v7);
  \draw[pairone]  (v0)--(v5) (v1)--(v3) (v2)--(v4) (v6)--(v7);
  \draw[pairtwo]  (v0)--(v1) (v2)--(v6) (v3)--(v5) (v4)--(v7);
  \pairnodes
\end{tikzpicture}
&
\begin{tikzpicture}
  \paircoords
  \draw[pairzero] (v0)--(v2) (v1)--(v6) (v3)--(v7) (v4)--(v5);
  \draw[pairone]  (v0)--(v3) (v1)--(v5) (v2)--(v7) (v4)--(v6);
  \draw[pairtwo]  (v0)--(v5) (v1)--(v7) (v2)--(v4) (v3)--(v6);
  \pairnodes
\end{tikzpicture}
&
\begin{tikzpicture}
  \paircoords
  \draw[pairzero] (v0)--(v6) (v1)--(v5) (v2)--(v3) (v4)--(v7);
  \draw[pairone]  (v0)--(v7) (v1)--(v4) (v2)--(v6) (v3)--(v5);
  \draw[pairtwo]  (v0)--(v2) (v1)--(v6) (v3)--(v7) (v4)--(v5);
  \pairnodes
\end{tikzpicture}
\\[1mm]
\pairtitle{\Gamma_0^Z}{3} & \pairtitle{\Gamma_1^Z}{3} &
\pairtitle{\Gamma_2^Z}{4} \\[1mm]
\begin{tikzpicture}
  \paircoords
  \draw[pairzero] (v0)--(v3) (v1)--(v2) (v4)--(v6) (v5)--(v7);
  \draw[pairone]  (v0)--(v2) (v1)--(v6) (v3)--(v7) (v4)--(v5);
  \draw[pairtwo]  (v0)--(v6) (v1)--(v5) (v2)--(v3) (v4)--(v7);
  \pairnodes
\end{tikzpicture}
&
\begin{tikzpicture}
  \paircoords
  \draw[pairzero] (v0)--(v5) (v1)--(v3) (v2)--(v4) (v6)--(v7);
  \draw[pairone]  (v0)--(v3) (v1)--(v5) (v2)--(v7) (v4)--(v6);
  \draw[pairtwo]  (v0)--(v7) (v1)--(v4) (v2)--(v6) (v3)--(v5);
  \pairnodes
\end{tikzpicture}
&
\begin{tikzpicture}
  \paircoords
  \draw[pairzero] (v0)--(v1) (v2)--(v6) (v3)--(v5) (v4)--(v7);
  \draw[pairone]  (v0)--(v5) (v1)--(v7) (v2)--(v4) (v3)--(v6);
  \draw[pairtwo]  (v0)--(v2) (v1)--(v6) (v3)--(v7) (v4)--(v5);
  \pairnodes
\end{tikzpicture}
\end{tabular}

\vspace{1mm}
\begin{tikzpicture}[baseline=-0.6ex]
  \draw[pairzero] (0,0)--(0.65,0);
  \node[anchor=west,font=\scriptsize] at (0.72,0) {0};
  \draw[pairone] (1.18,0)--(1.83,0);
  \node[anchor=west,font=\scriptsize] at (1.90,0) {1};
  \draw[pairtwo] (2.36,0)--(3.01,0);
  \node[anchor=west,font=\scriptsize] at (3.08,0) {2};
\end{tikzpicture}
\endgroup
\caption{Pairing graphs obtained from the array in
Example~\ref{ex:pair-partition-array}.  In \(\Gamma_j^X\), an edge from
\(M_{ij}\) has colour \(i\); in \(\Gamma_i^Z\), it has colour \(j\).
Colours \(0,1,2\) are shown as blue solid, red dashed, and green dotted
lines, respectively.}
\label{fig:example1-pairing-graphs}
\end{figure}

Along a cycle in a pairing graph, the paired-difference equations telescope,
so the signed sum of CPM exponents along the corresponding Tanner walk is
zero.

\begin{proposition}[Short cycles forced by pairing graphs]
\label{prop:pair-cycle-obstruction}
If \(\Gamma_j^X\) contains a graph cycle of length \(r\ge2\), then the Tanner
graph of \(H_X\) contains a cycle of length at most \(2r\).  The analogous
statement holds for \(\Gamma_i^Z\) and \(H_Z\).  Hence
\[
 g_{\mathrm T}(H_X)\le 2g(\Gamma_j^X),\qquad
 g_{\mathrm T}(H_Z)\le 2g(\Gamma_i^Z).
\]
\end{proposition}

\begin{proof}
Let \(c_0,c_1,\ldots,c_{r-1},c_r=c_0\) be a cycle in \(\Gamma_j^X\), with
\(\{c_t,c_{t+1}\}\in M_{i_tj}\).  The paired-difference equation
\eqref{eq:matchingconstraint} gives
\[
 e_{i_t,c_t}-e_{i_t,c_{t+1}}
 =d_{j,c_t}-d_{j,c_{t+1}}.
\]
Summing in \(t\) makes the right-hand side telescope to zero.  This is
Fossorier's zero-sum condition for the corresponding length-\(2r\) Tanner walk in
\(H_X\).  Consecutive edges have different colours because each colour class
is a matching, so this closed walk never immediately reverses an edge and
contains a simple Tanner
cycle of length at most \(2r\).  Interchanging \(e\) and \(d\) proves the
\(Z\)-side statement.
\end{proof}

This obstruction is inexpensive to test.  Each of the \(2J\) pairing graphs
has \(L\) vertices and \(JL/2\) edges, so ordinary graph traversal detects
parallel edges, triangles, or general cycles in
\(O(J^2L)\) total adjacency work, independent of \(P\).  It can therefore
reject a bad \(M\) before any finite-field exponent search.

A triangle-free \(J\)-regular pairing graph also imposes a parameter-level
restriction.

\begin{corollary}[Necessary test for girth eight]
\label{cor:girth-eight-prefilter}
If both CSS Tanner graphs have girth at least eight, every pairing graph is
simple, \(J\)-regular, and triangle-free.  In particular,
\[
 L\ge2J.
\]
\end{corollary}

\begin{proof}
Parallel edges and triangles would force Tanner cycles of length at most four
and six.  Each pairing graph has \(JL/2\) edges, whereas Mantel's bound gives
at most \(L^2/4\) edges in a triangle-free graph on \(L\) vertices.  Thus
\(JL/2\le L^2/4\), which gives \(L\ge2J\).
\end{proof}

Thus \(L<2J\) is a constant-time impossibility check for this girth-eight
design stage; no pairing or exponent enumeration should be attempted in
that range.

There is also an exponent-independent ceiling which does not use \(M\).
For two block rows and three block columns, the signed exponent sum in
Fossorier's length-12 cycle condition cancels identically.  The condition
therefore guarantees a Tanner cycle of length at most twelve, regardless of
the choice of \(M,E,D\).

\begin{proposition}[Universal girth-12 upper bound]
\label{prop:universal-twelve-obstruction}
For \(J\ge2\) and \(L\ge3\), every CPM lift of the complete
\(J\times L\) protograph has Tanner girth at most twelve.
\end{proposition}

\begin{proof}
Apply Fossorier's CPM cycle condition \cite{fossorier2004} to two distinct
block rows and three distinct block columns.
\end{proof}

Once \(J\ge2\) and \(L\ge3\) are known, this \(O(1)\) family-level test rules
out exponent searches with target girth greater than twelve.

After an array \(M\) passes these combinatorial tests, cycles that depend on
the exponent labels are checked using Fossorier's CPM cycle condition
\cite{fossorier2004}.  For example, block rows \(i,i'\) and block columns
\(\ell,\ell'\) form a Tanner 4-cycle in \(H_X\) exactly when
\begin{equation}
 (e_{i\ell}-e_{i\ell'})+(e_{i'\ell'}-e_{i'\ell})=0
 \quad\text{in }\F_P.
 \label{eq:four-cycle-test}
\end{equation}
The corresponding test with \(d\) is applied to \(H_Z\).  Thus girth at
least six requires \(2\binom{J}{2}\binom{L}{2}\) nonzero tests in total.
For target girth eight, the lifted 6-cycle conditions are also tested.  For
fixed \(J,L\), the number of these modular tests is independent of \(P\),
although their bit cost contains \(\log P\).

\section{Code Candidate Search}
\label{sec:construction-procedure}

A candidate is a triple \((M,E,D)\), where \(M=(M_{ij})\) is a
pair-partition array and
\(E=(e_{j\ell}),D=(d_{j\ell})\in\F_P^{J\times L}\).  The array \(M\) is
first required to pass the design tests in Section~\ref{sec:pair-partition-design}.
The exponents must then satisfy the paired-difference equations and the CPM
cycle tests.

\subsection{Equivalent exponent arrays}

If \((E,D)\) satisfies the paired-difference equations, then so does
\begin{equation}
 e'_{i\ell}=e_{i\ell}+\tau_\ell+\alpha_i,\qquad
 d'_{j\ell}=d_{j\ell}+\tau_\ell+\beta_j,
 \quad \tau_\ell,\alpha_i,\beta_j\in\F_P.
 \label{eq:equivalent-exponents}
\end{equation}
These changes relabel cyclic lift indices within block columns and block rows.
They preserve CSS orthogonality, binary ranks, Tanner girth, and distance.
The search may therefore fix a convenient representative as a gauge
convention for the exponent description.

\subsection{Solution space and candidate generation}

For fixed \(M\) and prime \(P\), Eq.~\eqref{eq:matchingconstraint} is a
homogeneous linear system in the \(2JL\) entries of \(E,D\).  We row-reduce
this system over \(\F_P\), obtain a basis for its solution space, and
enumerate or sample coefficient vectors within a prescribed budget.  An
exponent assignment is rejected when
the signed exponent sum for any 4-cycle is zero.  A girth-eight search also
rejects every lifted 6-cycle.

\subsection{Distance-guided forbidden patterns}

Suppose a rejected lift contains a low-weight logical operator with support
\[
 S=\{(\ell_t,a_t):1\le t\le w\}
 \subseteq\mathbb Z_L\times\mathbb F_P,
\]
where \(\ell_t\) is the block-column index and \(a_t\) is the column lift
index.  The common translation
\((\ell_t,a_t)\mapsto(\ell_t,a_t+c)\) preserves the CPM structure.  A
finite bank \(B\) therefore stores the block-column indices and relative
lift-index differences, such as \(a_t-a_1\), rather than an absolute
translate.

For a later exponent assignment with the same \(L,P\) and compatible
block-column labels, each stored pattern is tested against the current CPM
shifts; the new matching array may differ from the old one.  If some translated
support again has zero syndrome on the relevant CSS side, that assignment is rejected
before the full binary search.  A newly found low-weight logical operator is
converted to the same relative form and may be added to later banks.  Thus
observed failure mechanisms influence subsequent exponent choices.

The zero-syndrome test is especially simple for CPMs.  Normalize a stored support
as \(\{(\ell_t,r_t):1\le t\le w\}\), where \(r_1=0\).  For an \(X\)-type
pattern, define
\begin{equation}
 q_{j,b}(D):=
 \sum_{t=1}^{w}\mathbf 1\{r_t+d_{j\ell_t}=b\}\pmod 2,
 \qquad 0\le j<J,\quad b\in\F_P.
 \label{eq:forbidden-pattern-parity}
\end{equation}
A translate of this support has zero \(H_Z\)-syndrome exactly when
\(q_{j,b}(D)=0\) for every \(j,b\).  Indeed, the column
\((\ell_t,r_t+c)\) is incident to the check-row lift index
\(b=r_t+c+d_{j\ell_t}\); the common \(c\) merely shifts all check-row lift
indices.  For a \(Z\)-type pattern, the same test uses
\(e_{i\ell_t}\) in place of \(d_{j\ell_t}\).

\begin{proposition}[Universal cyclic automorphism of a CPM-CSS lift]
\label{prop:universal-cyclic-automorphism}
Let \(Q\) be the qubit-permutation matrix induced by the simultaneous lift
shift \(T:(\ell,a)\mapsto(\ell,a+1)\).  There are check-row permutation
matrices \(R_X,R_Z\) such that
\begin{equation}
 H_XQ=R_XH_X,\qquad H_ZQ=R_ZH_Z.
 \label{eq:cyclic-matrix-automorphism}
\end{equation}
Consequently, \(T\) is a color-preserving automorphism of both Tanner graphs
and of the CSS stabilizer.  In particular, the code automorphism group
contains the cyclic subgroup \(\langle T\rangle\cong C_P\).
\end{proposition}

\begin{proof}
In a CPM \(C(s)\), simultaneously increasing the column lift index by
one changes the row lift index \(b\) to \(b+1\).  The same cyclic row permutation
works in every CPM of a block row, independently of its exponent \(s\).
Applying it to all block rows gives \(R_X\) and \(R_Z\) in
Eq.~\eqref{eq:cyclic-matrix-automorphism}.  Thus the permutation preserves
the two check-row spaces separately.  Its order is \(P\), which gives the
stated cyclic subgroup.
\end{proof}

\begin{corollary}[Cyclic orbit of a logical operator]
\label{cor:cyclic-logical-orbit}
For either CSS side, if
\[
 v\in\ker H\setminus\rowspan(G),
 \qquad (H,G)=(H_Z,H_X)\ \text{or}\ (H_X,H_Z),
\]
then \(T^cv\in\ker H\setminus\rowspan(G)\) and
\(\wt(T^cv)=\wt(v)\) for every \(c\in\F_P\).  If \(P\) is prime and
\(0<\wt(v)<P\), these \(P\) binary vectors have pairwise distinct supports.
Their classes in the logical quotient \(\ker H/\rowspan(G)\) are determined
by the induced action of \(T\) and may coincide or be linearly dependent.
\end{corollary}

\begin{proof}
Proposition~\ref{prop:universal-cyclic-automorphism} shows that \(T\)
preserves both \(\ker H\) and \(\rowspan(G)\), as well as Hamming weight.
If \(T^cv\) belonged to \(\rowspan(G)\), applying \(T^{-c}\) would put
\(v\) in that row space, a contradiction.  Thus every translate is a
logical representative.

Because \(T\) has prime order \(P\), the orbit of a support has size one or
\(P\).  A support fixed by a nonidentity shift is fixed by all shifts; within
each block column it must therefore be empty or contain all \(P\) lift
indices.  Its weight is then a multiple of \(P\).  A nonzero support of
weight below \(P\) cannot be fixed, so its orbit has size \(P\).  Distinct
binary representatives can differ by elements of \(\rowspan(G)\), which
gives the stated conclusion about their logical classes.
\end{proof}

\begin{remark}[Logical action of the cyclic automorphism]
Proposition~\ref{prop:universal-cyclic-automorphism} gives a physical
qubit automorphism that preserves the \(X\)- and \(Z\)-check spaces
separately.  Its logical action, and whether translated supports define
distinct logical classes, are determined by its induced action on a logical
basis.
\end{remark}

\begin{example}[A 53-element orbit of an observed weight-12 logical]
\label{ex:forbidden-weight12}
In a \((J,L,P)=(3,10,53)\) search aimed at eliminating logical operators
through weight 12, one rejected trial produced an \(X\)-type logical with
normalized pairs
\begin{align*}
 (\ell_t)_{t=1}^{12}
   &=(0,5,4,1,3,6,3,4,1,8,8,5),\\
 (r_t)_{t=1}^{12}
   &=(0,2,3,5,8,12,20,26,39,41,43,52).
\end{align*}
For the rejected trial, Eq.~\eqref{eq:forbidden-pattern-parity} vanished for
all three \(Z\)-check block rows, and direct row-space reduction confirmed
that the vector was not in \(\rowspan(H_X)\).  The search therefore stored
the two displayed lists and the side label \(X\).  For every later exponent
array \(D\), it forms the three parity tables \(q_{j,b}(D)\).  If all their
entries vanish, the same relative support, together with all 53 of its cyclic
translates, again has zero syndrome, so the candidate is rejected.  Every
retained candidate must instead have \(q_{j,b}(D)=1\) for at least one pair
\((j,b)\); the observed support then has nonzero syndrome and cannot reappear
as an \(X\)-type logical operator.

Since \(12<53\), Corollary~\ref{cor:cyclic-logical-orbit} shows that the
rejected lift actually contains 53 distinct weight-12 binary supports in this
cyclic orbit.  Equation~\eqref{eq:forbidden-pattern-parity} is invariant under
their common translation: for a later lift, either all 53 translates have
zero syndrome or none does.  One normalized bank entry therefore excludes
the whole orbit at once.

\end{example}

\subsection{Complete screening loop}

Stored patterns are used only to screen candidates; every reported distance
lower bound is recomputed on the retained fixed matrix by the complete search
of Section~\ref{sec:distance-certification-section}.  In
Algorithm~\ref{alg:candidate-screening}, the cutoff \(W\) is a design target.
The tests are ordered by cost: pairing graphs first, modular cycle tests and
stored patterns second, and binary rank, girth, and logical checks last.

\begin{algorithm}[H]
\small
\caption{PP code-candidate screening}
\label{alg:candidate-screening}
\begin{algorithmic}[1]
\Require \(J,L,P\), target girth \(g\in\{6,8\}\), search cutoff \(W\),
an exponent-search budget, and a forbidden-pattern bank \(B\)
\Ensure Fixed \(M,E,D,H_X,H_Z\) with verification records, or no candidate
within the budget
\State Generate an \(M\) satisfying the pairing-graph tests of
Section~\ref{sec:pair-partition-design}
\State Solve Eq.~\eqref{eq:matchingconstraint} over \(\F_P\)
\For{each exponent pair \((E,D)\) generated within the budget}
  \If{the signed CPM-exponent sum is zero for a cycle of length less than \(g\)}
    \State reject \((E,D)\)
  \ElsIf{a translated support recorded in \(B\) has zero syndrome for the current shifts}
    \State reject \((E,D)\)
  \Else
    \State build \(H_X,H_Z\) and verify CSS orthogonality, ranks, and girths
    \State run the direct logical search through weight \(W\) on both sides
    \If{a logical operator of weight at most \(W\) is found}
      \State store its relative support for possible use in later banks
      \State reject \((E,D)\)
    \Else
      \State return the fixed matrices and their verification records
    \EndIf
  \EndIf
\EndFor
\State report that no candidate was found within the prescribed budget
\end{algorithmic}
\end{algorithm}

Let \(\mathcal N\) be the exponent-search budget.  For each sampled exponent
pair, all 4-cycle tests cost
\(O(J^2L^2)\); a direct enumeration of 6-cycle types for target girth eight
is \(O(J^3L^3)\).  Building both sparse lifts writes \(O(JLP)\) nonzeros.
Thus the screening cost is at most \(\mathcal N\) times these per-candidate
tests plus binary rank and distance verification.  The latter can dominate:
its support tree is exponential in \(W\), although Proposition~
\ref{prop:p-independent-search-tree} removes \(P\) from that tree bound when
\(J,L,W\) are fixed.  Matrix preprocessing still depends on \(P\).

\section{Distance Verification}
\label{sec:distance-certification-section}

Algorithm~\ref{alg:bounded-logical-search} exhaustively searches
\(\ker H\setminus\rowspan G\) by branching from unsatisfied checks and using
proved lower bounds for pruning.  A completed \textsc{None} run certifies
\(d(H,G)>W\).  For CPM matrices, cyclic translation reduces the root set from
\(LP\) positions to \(L\).  For fixed \(J,L,W\), the resulting combinatorial
support-enumeration bound is independent of \(P\), although matrix
preprocessing and row-space membership tests can still grow with \(P\).  The
accompanying run records identify the fixed matrices and completed searches.

For a binary matrix \(H\), a nonzero vector \(v\) with \(Hv=0\) is a
zero-syndrome vector, and the classical minimum distance is the least weight
of such a vector.  Searches for small stopping sets and their codewords were
studied in \cite{richter2006,rosnes2009}; a search specialized to structured
CPM-LDPC codes was given in \cite{khatami2013}.  In the CSS setting one must
add the stabilizer quotient: an \(X\)-logical vector lies in \(\ker H_Z\)
but outside \(\rowspan H_X\), and a \(Z\)-logical vector is defined with
\(X\) and \(Z\) interchanged.

\subsection{Distance lower-bound verification}
\label{sec:distance-certification}
Let \(H\) and \(G\) be binary matrices with
\(\rowspan G\subseteq\ker H\), and define
\begin{equation}
 d(H,G)=\min\{\wt x:Hx=0,\ x\notin\rowspan G\}.
 \label{eq:quotient-distance}
\end{equation}
For an \(X\)-logical search use \((H,G)=(H_Z,H_X)\); for a \(Z\)-logical
search use \((H_X,H_Z)\).  Thus the quantum distance is the smaller of the
two values in Eq.~\eqref{eq:quotient-distance}.

Algorithm~\ref{alg:bounded-logical-search} searches all relevant supports up
to a weight limit \(W\).  A state consists of a selected support \(S\), a locally
forbidden set \(F\), and its syndrome \(\sigma(S)=H\mathbf1_S^{\mathsf T}\).
At a high level, it grows a connected candidate only where an unsatisfied
check says that another variable is necessary.  This avoids enumerating all
\(\binom{LP}{w}\) supports of every weight \(w\), while retaining a branch
leading to every minimum logical operator.
For an unsatisfied check \(c\), let
\[
 A_c(S,F)=N(c)\setminus(S\cup F)
\]
be its available variables, additionally restricted to indices larger than
the root when the minimum-root convention is used.  Let \(\Delta\) be the
maximum column weight of \(H\).  One valid lower bound on the number of
additional variables is
\[
 b_1=\left\lceil\frac{|\sigma(S)|}{\Delta}\right\rceil.
\]
Another, \(b_2\), is the size of any family of unsatisfied checks whose
sets \(A_c(S,F)\) are pairwise disjoint.
The implementation obtains \(b_2\) greedily and uses \(b=\max\{b_1,b_2\}\).

\begin{algorithm}[H]
\small
\caption{Complete search up to a weight limit}
\label{alg:bounded-logical-search}
\begin{algorithmic}[1]
\Require \(H,G\), weight limit \(W\), and roots \(R\)
\Ensure \(x\in\ker H\setminus\rowspan G\) with \(\wt x\le W\), or \textsc{None}
\For{each root \(r\in R\)}
  \State start depth-first search from \(S=\{r\}\) and \(F=\varnothing\)
  \If{a vector is returned} \State return it \EndIf
\EndFor
\State return \textsc{None}
\Statex
\State At a depth-first-search state \((S,F)\):
\If{\(\sigma(S)=0\)}
  \If{\(\mathbf1_S\notin\rowspan G\)} \State return \(\mathbf1_S\) \EndIf
  \State return \textsc{Fail} from this branch
\EndIf
\State choose an unsatisfied check \(c\) with minimum \(|A_c(S,F)|\)
\If{\(A_c(S,F)=\varnothing\) or \(|S|+b>W\)}
  \State return \textsc{Fail} from this branch
\EndIf
\State set \(F_{\mathrm{loc}}=F\)
\For{each \(v\in A_c(S,F)\) in a fixed order}
  \State search recursively from \((S\cup\{v\},F_{\mathrm{loc}})\)
  \If{a vector is returned} \State return it \EndIf
  \State \(F_{\mathrm{loc}}\gets F_{\mathrm{loc}}\cup\{v\}\)
\EndFor
\State return \textsc{Fail} from this branch
\end{algorithmic}
\end{algorithm}

For a generic matrix, one may take every variable as a possible minimum
root.  For a CPM matrix, label variables by
\((\ell,a)\in\mathbb Z_L\times\mathbb F_P\), ordered by \(\ell P+a\), and
take only \(R=\{(\ell,0):\ell\in\mathbb Z_L\}\).  The common cyclic
translation \((\ell,a)\mapsto(\ell,a+c)\) preserves both \(H\) and
\(\rowspan G\).  The implementation that records the complete run additionally requires prime
\(P\) and \(W<P\) before enabling this \(L\)-root reduction.

The next proposition bounds combinatorial support enumeration separately
from linear algebra on the lifted matrices.

\begin{proposition}[Lift-size-independent support enumeration]
\label{prop:p-independent-search-tree}
For a fully populated CPM matrix of row weight \(L\), fix \(J,L,W\), use the
\(L\) cyclic roots above, and stop at weight \(W\).  The number of
depth-first-search support states visited by
Algorithm~\ref{alg:bounded-logical-search} is at most
\begin{equation}
 L\sum_{t=0}^{W-1}(L-1)^t .
 \label{eq:p-independent-tree-bound}
\end{equation}
In particular, this combinatorial bound is independent of the lift size
\(P\).
\end{proposition}

\begin{proof}
Each search starts from one selected root.  At a nonzero-syndrome state, the
chosen unsatisfied check contains an odd, hence nonzero, number of selected
variables.  Since that check has row weight \(L\), at most \(L-1\) of its
variables remain available for branching.  Each recursive step adds one
variable, so there are at most \(W-1\) branching levels after the root.
Multiplying the geometric-tree bound by the \(L\) cyclic roots
gives~\eqref{eq:p-independent-tree-bound}.  Pruning can only reduce the
number of visited states.
\end{proof}

Thus this worst-case support-enumeration tree remains unchanged as \(P\)
increases with \(J,L,W\) fixed.  Constructing sparse matrices of length
\(LP\), computing ranks or row-echelon data, updating syndromes, and testing
membership in \(\rowspan G\), however, can depend on \(P\).  If \(W\) grows
with \(P\), then Eq.~\eqref{eq:p-independent-tree-bound} also inherits that
dependence through \(W\).

\begin{lemma}[Structure of a minimum logical support]
\label{lem:minimum-logical-support}
If \(d(H,G)<\infty\), there is a minimum-weight
\(x\in\ker H\setminus\rowspan G\) whose support is Tanner-connected.
Moreover, its support contains no proper nonempty support of a vector in
\(\ker H\).
\end{lemma}

\begin{proof}
Split the support of a minimum-weight \(x\) into its Tanner-connected
components.  Each component separately has zero syndrome, because no check
is adjacent to selected variables in two different components.  If every
component belonged to \(\rowspan G\), then their sum \(x\) would also belong
to \(\rowspan G\).  Hence one component is outside \(\rowspan G\), and
minimum weight forces it to be the whole support.

Now suppose that \(0\ne y\in\ker H\) and
\(\operatorname{supp}y\subsetneq\operatorname{supp}x\).  Then
\(x+y\in\ker H\), and both \(y\) and \(x+y\) have smaller weight than
\(x\).  They cannot both lie in \(\rowspan G\), because their sum is \(x\).
Thus one would be a lower-weight member of
\(\ker H\setminus\rowspan G\), a contradiction.
\end{proof}

\begin{theorem}[Correctness of the distance lower-bound search]
\label{thm:distance-search-correctness}
Let \(H,G\) be binary matrices with
\(\rowspan G\subseteq\ker H\), and let \(W\ge1\).  Assume that every
nonempty support of weight at most \(W\) has, under a coordinate symmetry
preserving both \(\ker H\) and \(\rowspan G\), a representative whose least
variable is in the root set \(R\).
Algorithm~\ref{alg:bounded-logical-search} returns \textsc{None} if and only
if \(d(H,G)>W\).  Consequently, running it on both CSS sides and obtaining
\textsc{None} proves quantum distance \(d>W\).

For the CPM matrices in this paper, the common cyclic translations give the
required coverage with
\(R=\{(\ell,0):\ell\in\mathbb Z_L\}\).  Hence only \(L\) root searches,
rather than \(LP\), are needed for a complete lower-bound certificate.
\end{theorem}

\begin{proof}
Every returned vector has zero syndrome, lies outside \(\rowspan G\), and has
weight at most \(W\), so it is valid.

To show that no valid vector is missed, suppose \(d(H,G)\le W\), and choose
the minimum-weight representative \(x\) from
Lemma~\ref{lem:minimum-logical-support}.

Choose a symmetry-equivalent support whose root is the minimum variable.  In the
CPM case, take the least block-column index \(\ell_0\) in the support and
translate one selected \((\ell_0,a_0)\) to \((\ell_0,0)\).  This variable is
then the minimum root, proving that the \(L\) CPM roots cover \(x\).

Consider a search state with \(S\subseteq\operatorname{supp}x\) and
\(F\cap\operatorname{supp}x=\varnothing\).  If \(S\) is proper, its syndrome
cannot be zero by the preceding minimality argument.  Let \(c\) be the
unsatisfied check selected by the algorithm.  The parity of
\(\operatorname{supp}x\) at \(c\) is even and that of \(S\) is odd, so at
least one variable of \(\operatorname{supp}x\setminus S\) lies in
\(A_c(S,F)\).  In the ordered loop, the first such variable is explored
before it can enter \(F_{\mathrm{loc}}\), giving another state with
\(S\subseteq\operatorname{supp}x\) and
\(F\cap\operatorname{supp}x=\varnothing\).

Such a state is never stopped early because \(x\) extends \(S\) to a valid
vector.  Each future variable can repair at most \(\Delta\)
currently unsatisfied checks, proving \(b_1\); checks with disjoint available
sets require distinct future variables, proving \(b_2\).  Hence
\(|S|+b\le\wt x\le W\).  Repeating the argument reaches the full support of
\(x\), where the algorithm returns it.  This contradicts \textsc{None} and
proves completeness.
\end{proof}

A completed \textsc{None} result is a certificate of exhaustive support
enumeration through \(W\).  Its wall-clock cost includes the \(P\)-dependent
matrix operations described after Proposition~\ref{prop:p-independent-search-tree}.

For every fully populated CPM block matrix here, the \(P\) rows in any one
block-row family sum to the all-one vector.  Hence \(Hv=0\) implies
\begin{equation}
 0=\mathbf1^{\mathsf T}v=\wt(v)\pmod2.
 \label{eq:even-zero-syndrome-weight}
\end{equation}
All kernel vectors therefore have even weight.  For an even target \(d_0\),
two completed searches through \(W=d_0-2\) prove \(d\ge d_0\).  A separately
verified non-stabilizer zero-syndrome vector of weight \(d_0\) supplies the
matching upper bound and proves equality.

The lower-bound JSON identifies the input matrix and records the weight limit,
result for each root, and completion status.  It is a reproducible run record: independent
verification reruns the proved algorithm on the identified input.  In
contrast, each reported upper-bound vector is checked directly for zero
syndrome and nonmembership in the opposite stabilizer row space.

\subsection{Factorial upper bound and explicit witnesses}

The complete search gives instance-specific lower bounds.  Independently,
the fact that every block is a CPM gives a universal upper bound: maximal minors of either
polynomial parity-check matrix produce low-weight kernel vectors.  A counting
argument ensures that at least one of them is logical rather than a
stabilizer.

Let \(R_P=\F_2[x]/(x^P-1)\), and represent the CPM matrices by

\[
 \widehat H_X(x)=(x^{e_{j\ell}})_{j,\ell},\qquad
 \widehat H_Z(x)=(x^{d_{j\ell}})_{j,\ell}
 \in R_P^{J\times L}.
\]

For a set \(S\subseteq\Z_L\) of \(J+1\) block columns, define
\(c_Z^S(x)\in R_P^L\) by
\begin{equation}
 (c_Z^S(x))_\ell=
 \begin{cases}
 \operatorname{perm}\!\left(\widehat H_Z(x)_{[:,S\setminus\{\ell\}]}\right),
   & \ell\in S,\\
 0,&\ell\notin S.
 \end{cases}
 \label{eq:permanent-vector}
\end{equation}
Over characteristic two, the usual maximal-minor identity gives
\(\widehat H_Z(x)c_Z^S(x)^{\mathsf T}=0\).  Each of its \(J+1\) nonzero
coordinates contains at most \(J!\) polynomial terms, so its binary Hamming weight
is at most \((J+1)!\) \cite{smarandache2012}.

The two assumptions in the theorem below are natural for the CPM family
considered here.  First, the elementary parameter count gives \(LP\) qubits
and \(2JP\) check rows.  Requiring \(LP-2JP>0\) gives \(L>2J\), or,
for integer \(L\), \(L\ge2J+1\).  This places the construction in the
positive-rate parameter range before rank deficiencies are counted.

Second, each fully populated \(J\times L\) CPM block matrix has \(JP\) rows.
Within any one block-row family, the sum of its \(P\) rows is the all-one
vector of length \(LP\), because every block is a permutation matrix.
The \(J\) block-row families therefore give \(J-1\) unavoidable independent
row dependencies, and hence
\[
 \rank_{\F_2}H_X,\ \rank_{\F_2}H_Z\le JP-(J-1)=J(P-1)+1.
\]
Equality means that there are no additional row dependencies beyond these
CPM-forced ones.  It is therefore the natural maximum-rank condition, and it
is directly verified for every code in Table~\ref{tab:finite-summary}.

\begin{theorem}[Factorial upper bound on quantum distance]
\label{thm:distance-upper-bound}
Let \(H_X,H_Z\) be the CPM matrices of a PP code with \(J\ge2\), odd \(P\),
and
\[
 L\ge2J+1,\qquad
 \rank_{\F_2}H_X=\rank_{\F_2}H_Z=J(P-1)+1.
\]
Then there exist
\[
 v_X\in\ker H_Z\setminus\rowspan H_X,\qquad
 v_Z\in\ker H_X\setminus\rowspan H_Z
\]
with
\[
 \wt v_X\le(J+1)!,\qquad \wt v_Z\le(J+1)!.
\]
Consequently,
\[
 d_X\le(J+1)!,\qquad d_Z\le(J+1)!,\qquad d\le(J+1)!.
\]
\end{theorem}

\begin{proof}
At \(x=1\), both polynomial parity-check matrices become the all-one
\(J\times L\) matrix and have rank one.  In the Chinese-remainder
decomposition of \(R_P\), every nontrivial component has row rank at most
\(J\).  Since \(J(P-1)+1\) is the largest possible total binary rank, the
assumed equality forces \(\widehat H_X(x)\) and \(\widehat H_Z(x)\) to have
row rank \(J\) in every nontrivial component.

For every \((J+1)\)-subset \(S\) of block columns, form the maximal-minor
vector \(c_Z^S(x)\) in Eq.~\eqref{eq:permanent-vector}.  Include all its
common cyclic shifts, and let \(\mathcal C_Z\) be their binary span.
Over the field associated with any nontrivial irreducible factor,
\(\widehat H_Z(x)\) has row rank \(J\).  Standard linear algebra then shows
that its maximal-minor circuit vectors span its kernel, whose dimension is
\(L-J\).  At the factor \(x=1\), every polynomial term evaluates to one and every
\(J\times J\) permanent evaluates to \(J!=0\) in \(\F_2\), since \(J\ge2\).
The Chinese-remainder decomposition therefore gives
\[
 \dim_{\F_2}\mathcal C_Z=(L-J)(P-1).
\]
Thus \(\mathcal C_Z\) contains \(2^{(L-J)(P-1)}\) vectors, whereas
\(\rowspan H_X\) contains \(2^{\rank H_X}\) vectors.  Moreover,
\[
 (L-J)(P-1)\ge(J+1)(P-1)>J(P-1)+1=\rank H_X,
\]
where the strict inequality uses \(P\ge3\).  Hence
\(\mathcal C_Z\not\subseteq\rowspan H_X\).

Simultaneously shifting the \(P\) lift indices within every block column
permutes the rows within each CPM block-row family.  Therefore
\(\rowspan H_X\) is invariant under this shift.  If every permanent
generator \(c_Z^S(x)\) were in \(\rowspan H_X\), then all its shifts and
their entire span would also be in that row space, a contradiction.  Hence
at least one permanent generator is a nonzero element of
\(\ker H_Z\setminus\rowspan H_X\), and its weight is at most \((J+1)!\).
This gives \(v_X\).  Interchanging \(X\) and \(Z\) gives \(v_Z\).
\end{proof}

\begin{corollary}[Exact distance from matching bounds]
\label{cor:matching-distance-bounds}
If the complete search proves \(d_X>w_X-1\) and an \(X\)-logical vector of
weight \(w_X\) is given, then \(d_X=w_X\).  The analogous statement holds
for \(d_Z\).  Hence, if matching lower and upper bounds are obtained on both
sides, then \(d=\min\{w_X,w_Z\}\).
\end{corollary}

Because different terms in a permanent may cancel, the theorem gives the upper
bound \((J+1)!\), rather than an equality at that weight.
All codes in Table~\ref{tab:finite-summary} satisfy the rank and size
conditions.  It gives the universal bound \(24\) for \(J=3\) and \(120\)
for \(J=4\).  The sharper exact upper bounds in the table come from the
accompanying explicit vectors.  Their zero syndrome and nonmembership in the
opposite stabilizer row space are checked directly.  In particular,
weight-20 vectors on both CSS sides provide the upper bound for
\(\qcode{944}{478}{20}\).

\section{Numerical Results}
\label{sec:numerical-results}

Table~\ref{tab:finite-summary} lists 36 independently certified CPM-based CSS
codes and three daggered zero-design-rate examples with contributor-reported
distances.  For each unmarked entry, binary ranks, Tanner girths, complete
lower bounds, and matching logical witnesses were checked.
The three smallest distance-12 codes were provided by Nishad Maskara and
independently checked here.

Jong Yeon Lee supplied 19 additional CPM-PP instances.  These instances were
constructed using the logical-spectroscopy method introduced in
\cite{lee2026LogicalSpectroscopy}.  The daggered \((J,L)=(5,10)\) rows are
\(\qcode{110}{8}{12}\) (three realizations), \(\qcode{170}{8}{20}\), and
\(\qcode{190}{8}{18}\): although \(1-2J/L=0\), rank dependencies leave
\(k=8\).  Their structural parameters were independently checked, while the
exact distances remain contributor reported.  The further reported \(d>L\)
codes \(\qcode{184}{50}{10}\), \(\qcode{488}{126}{16}\), and
\(\qcode{904}{230}{20}\), and the girth-eight realizations
\(\qcode{248}{66}{12}\), \(\qcode{296}{78}{12}\), are recorded on the
project webpage with their available verification data.  Table~\ref{tab:finite-summary}
distinguishes independently certified distances from contributor-reported
ones by the dagger notation.  For
\(\qcode{740}{378}{d}\),
verified weight-20 logical operators and the contributor-reported lower bound
give \(16\leq d\leq20\).  The webpage contains the full code data and
verification records \cite{okadaKasaiData2026}; the accompanying software is
available in the GitHub repository \cite{okadaKasaiCode2026}.

\begin{table}[H]
\centering
\scriptsize
\setlength{\tabcolsep}{3.2pt}
\renewcommand{\arraystretch}{0.88}
\caption{CPM-based CSS codes.  Daggers mark contributor-reported distances.}
\label{tab:finite-summary}
\begin{tabular}{@{}ccccc@{}}
\toprule
\(\qcode{n}{k}{d}\) & girth & \((J,L)\) & \(P\) & Rate \(k/n\)\\
\midrule
\(\qcode{232}{62}{12}\) & \(6\) & \((3,8)\) & \(29\) & \(0.267\)\\
\(\qcode{248}{66}{12}\) & \(6\) & \((3,8)\) & \(31\) & \(0.266\)\\
\(\qcode{296}{78}{12}\) & \(6\) & \((3,8)\) & \(37\) & \(0.264\)\\
\(\qcode{472}{122}{14}\) & \(6\) & \((3,8)\) & \(59\) & \(0.258\)\\
\(\qcode{472}{122}{16}\) & \(8\) & \((3,8)\) & \(59\) & \(0.258\)\\
\(\qcode{488}{126}{14}\) & \(6\) & \((3,8)\) & \(61\) & \(0.258\)\\
\(\qcode{584}{150}{18}\) & \(8\) & \((3,8)\) & \(73\) & \(0.257\)\\
\(\qcode{1112}{282}{20}\) & \(8\) & \((3,8)\) & \(139\) & \(0.254\)\\
\(\qcode{1336}{338}{22}\) & \(8\) & \((3,8)\) & \(167\) & \(0.253\)\\
\(\qcode{1784}{450}{24}\) & \(8\) & \((3,8)\) & \(223\) & \(0.252\)\\
\(\qcode{470}{192}{12}\) & \(6\) & \((3,10)\) & \(47\) & \(0.409\)\\
\(\qcode{530}{216}{12}\) & \(6\) & \((3,10)\) & \(53\) & \(0.408\)\\
\(\qcode{590}{240}{12}\) & \(6\) & \((3,10)\) & \(59\) & \(0.407\)\\
\(\qcode{710}{288}{14}\) & \(6\) & \((3,10)\) & \(71\) & \(0.406\)\\
\(\qcode{970}{392}{16}\) & \(6\) & \((3,10)\) & \(97\) & \(0.404\)\\
\(\qcode{1390}{560}{18}\) & \(6\) & \((3,10)\) & \(139\) & \(0.403\)\\
\(\qcode{1490}{600}{18}\) & \(6\) & \((3,10)\) & \(149\) & \(0.403\)\\
\(\qcode{1630}{656}{20}\) & \(8\) & \((3,10)\) & \(163\) & \(0.402\)\\
\(\qcode{2110}{848}{22}\) & \(6\) & \((3,10)\) & \(211\) & \(0.402\)\\
\(\qcode{2230}{896}{24}\) & \(8\) & \((3,10)\) & \(223\) & \(0.402\)\\
\(\qcode{1356}{682}{14}\) & \(6\) & \((3,12)\) & \(113\) & \(0.503\)\\
\(\qcode{1524}{766}{14}\) & \(6\) & \((3,12)\) & \(127\) & \(0.503\)\\
\(\qcode{2004}{1006}{16}\) & \(6\) & \((3,12)\) & \(167\) & \(0.502\)\\
\(\qcode{2676}{1342}{18}\) & \(6\) & \((3,12)\) & \(223\) & \(0.502\)\\
\(\qcode{1414}{812}{12}\) & \(6\) & \((3,14)\) & \(101\) & \(0.574\)\\
\(\qcode{3122}{1788}{16}\) & \(6\) & \((3,14)\) & \(223\) & \(0.573\)\\
\midrule
\(\qcode{276}{98}{14}\) & \(6\) & \((4,12)\) & \(23\) & \(0.355\)\\
\(\qcode{372}{130}{16}\) & \(6\) & \((4,12)\) & \(31\) & \(0.349\)\\
\(\qcode{444}{154}{18}\) & \(6\) & \((4,12)\) & \(37\) & \(0.347\)\\
\(\qcode{492}{170}{20}\) & \(6\) & \((4,12)\) & \(41\) & \(0.346\)\\
\(\qcode{516}{178}{20}\) & \(6\) & \((4,12)\) & \(43\) & \(0.345\)\\
\(\qcode{708}{242}{22}\) & \(6\) & \((4,12)\) & \(59\) & \(0.342\)\\
\(\qcode{518}{228}{16}\) & \(6\) & \((4,14)\) & \(37\) & \(0.440\)\\
\(\qcode{574}{252}{18}\) & \(6\) & \((4,14)\) & \(41\) & \(0.439\)\\
\(\qcode{848}{430}{18}\) & \(6\) & \((4,16)\) & \(53\) & \(0.507\)\\
\(\qcode{944}{478}{20}\) & \(6\) & \((4,16)\) & \(59\) & \(0.506\)\\
\midrule
\multicolumn{5}{c}{Zero-design-rate contributed examples}\\
\(\qcode{110}{8}{12}^{\dagger}\) & \(6\) & \((5,10)\) & \(11\) & \(0.073\)\\
\(\qcode{170}{8}{20}^{\dagger}\) & \(6\) & \((5,10)\) & \(17\) & \(0.047\)\\
\(\qcode{190}{8}{18}^{\dagger}\) & \(6\) & \((5,10)\) & \(19\) & \(0.042\)\\
\bottomrule
\end{tabular}
\end{table}

\section{Extension to APM-PP Codes}
\label{sec:apm-pp-extension}

CPM exponents make construction, cycle screening, and distance verification
especially transparent, but CSS cancellation only uses equality of relative
maps.  Pairing equal relative affine maps therefore extends the construction
to affine permutation matrices (APMs).  A structured specialization below
reproduces the selected check matrices, called the active check matrices in
the Kasai construction.

Let \(h(x)=ax+b\) be an affine permutation of \(\Z_P\), where
\(a\in\Z_P^\times\), and write \(A(h)\) for its permutation matrix.  In
parallel with the CPM exponent arrays \(E\) and \(D\) in Section~2, choose
affine maps \(e_{i\ell}\) and \(d_{j\ell}\), and set
\begin{equation}
 H_X=\bigl(A(e_{i\ell})\bigr)_{i,\ell},\qquad
 H_Z=\bigl(A(d_{j\ell})\bigr)_{j,\ell}.
 \label{eq:apm-blockmatrices}
\end{equation}
For a pair \((i,j)\) of an \(X\)-block row and a \(Z\)-block row, define
the relative map
\begin{equation}
 \rho_{ij\ell}=d_{j\ell}^{-1}\circ e_{i\ell}.
 \label{eq:apm-relative-map}
\end{equation}

\begin{definition}[APM-PP presentation]
An APM-PP presentation consists of the two affine-map arrays and, for every
block-row pair \((i,j)\), a pair partition \(M_{ij}\) of the \(L\) block
columns such that
\begin{equation}
 \rho_{iju}=\rho_{ijv}
 \quad\text{for every }\{u,v\}\in M_{ij}.
 \label{eq:apm-pp-condition}
\end{equation}
\end{definition}
When all affine multipliers are one, this condition reduces, up to the CPM
sign convention, to the paired-difference condition of Section~2.

\begin{theorem}[APM-PP construction]
\label{thm:apm-pp-construction}
Every APM-PP presentation satisfies \(H_XH_Z^{\mathsf T}=0\).  If both
arrays have \(J\) block rows, the resulting regular CSS code has length
\(LP\), row weight \(L\), column weight \(J\), and dimension
\[
 k=LP-\rank H_X-\rank H_Z.
\]
\end{theorem}

\begin{proof}
An \(X\)-check row with lift index \(r\) in block row \(i\) and a
\(Z\)-check row with lift index \(s\) in block row \(j\) meet in block
column \(\ell\) if and only if \(s=\rho_{ij\ell}(r)\).
Equation~\eqref{eq:apm-pp-condition} makes this condition simultaneously
true or false at the two endpoints of every pair.  Hence every \(X\)--\(Z\)
row intersection is even.  The other claims follow because every nonzero block
is a permutation matrix.
\end{proof}

The Kasai construction starts from \(L=2q\) block columns, divided into two
halves, and two lists of \(P\times P\) affine permutation matrices
\(F_u,G_u\), \(u\in\Z_q\).
Cyclically shifting their indices gives the \(q\)-block-row parent matrices
\begin{equation}
 \widehat H_X=
 \left(\,(F_{t-i})_{i,t\in\Z_q}\ \middle|\
             (G_{t-i})_{i,t\in\Z_q}\,\right),\qquad
 \widehat H_Z=
 \left(\,(G_{j-t}^{\mathsf T})_{j,t\in\Z_q}\ \middle|\
             (F_{j-t}^{\mathsf T})_{j,t\in\Z_q}\,\right).
 \label{eq:kasai-parent-matrices}
\end{equation}
Choose \(S\subseteq\Z_q\), \(|S|=J\), and let \(H_X\) and \(H_Z\) be the
block-row submatrices of \(\widehat H_X\) and \(\widehat H_Z\), respectively,
with row indices in \(S\).  The Kasai terminology calls these selected rows
active and the complementary parent rows latent; below we use selected rows
and complementary parent rows.  The set
\begin{equation}
 \Delta_S=\{j-i\bmod q:i,j\in S\}
 \label{eq:kasai-selected-differences}
\end{equation}
contains exactly the row differences occurring between selected block rows.

Equivalently, the blocks of the selected matrices are
\begin{align}
 (H_X)_{i,t}&=F_{t-i},& (H_X)_{i,q+t}&=G_{t-i},\label{eq:kasai-apm-x-blocks}\\
 (H_Z)_{j,t}&=G_{j-t}^{\mathsf T},& (H_Z)_{j,q+t}&=F_{j-t}^{\mathsf T},
 \label{eq:kasai-apm-z-blocks}
\end{align}
for \(i,j\in S\) and \(t\in\Z_q\).  If
\begin{equation}
 F_uG_{r-u}=G_{r-u}F_u
 \quad(r\in\Delta_S,\ u\in\Z_q),
 \label{eq:kasai-active-commutativity}
\end{equation}
then every selected block-row pair has the APM-PP certificate
\begin{equation}
 M_{ij}=\bigl\{\{\,i+u,\ q+j-u\,\}:u\in\Z_q\bigr\}.
 \label{eq:kasai-apm-matching}
\end{equation}
Indeed, the two relative maps in a displayed pair are
represented by \(F_uG_{j-i-u}\) and \(G_{j-i-u}F_u\), respectively, and hence
are equal by Eq.~\eqref{eq:kasai-active-commutativity}.

The \((i,j)\) block of the parent product depends only on \(r=j-i\) and is
\begin{equation}
 (\widehat H_X\widehat H_Z^{\mathsf T})_{i,j}=\Psi_{j-i},\qquad
 \Psi_r:=\sum_{u\in\Z_q}
 \bigl(F_uG_{r-u}+G_{r-u}F_u\bigr).
 \label{eq:kasai-delta-psi}
\end{equation}
For \(r\in\Delta_S\), Eq.~\eqref{eq:kasai-active-commutativity} makes
\(\Psi_r=0\), which gives CSS orthogonality of the selected check matrices.
The full Kasai parent construction retains the complementary rows and requires
\(\Psi_r\ne0\) for some \(r\notin\Delta_S\), so that the entire parent pair is
not orthogonal.  Thus the specialization identifies the APM-PP check matrices
with the selected check matrices of the Kasai construction, while
Eq.~\eqref{eq:kasai-delta-psi} supplies the additional parent-level condition.

For a concrete \((J,L)=(3,12)\) example, take \(q=6\) and
\(S=\{0,1,2\}\).  Suppressing the other parent block rows, the zeroth rows of
Eq.~\eqref{eq:kasai-parent-matrices} are
\begin{align*}
 (\widehat H_X)_{0,*}
 &=\bigl(F_0,F_1,\ldots,F_5\mid G_0,G_1,\ldots,G_5\bigr),\\
 (\widehat H_Z)_{0,*}
 &=\bigl(G_0^{\mathsf T},G_5^{\mathsf T},\ldots,G_1^{\mathsf T}\mid
          F_0^{\mathsf T},F_5^{\mathsf T},\ldots,F_1^{\mathsf T}\bigr).
\end{align*}
The code matrices \(H_X,H_Z\) use parent block rows \(0,1,2\).  Then
\(\Delta_S=\{0,1,2,4,5\}\), and Eq.~\eqref{eq:kasai-apm-matching} gives, for
example,
\[
 M_{00}=\{\{0,6\},\{1,11\},\{2,10\},\{3,9\},\{4,8\},\{5,7\}\}.
\]
Commutativity is imposed for \(r\in\Delta_S\), while the additional
complementary-parent-row requirement can be placed at the missing residue as
\(\Psi_3\ne0\).  Here \(M_{00}\) specifies the block-row-dependent matching,
while \(\Psi_3\ne0\) supplies the complementary-parent constraint.

The complete distance lower-bound test also applies to a fixed APM-PP lift
because Theorem~\ref{thm:distance-search-correctness} is stated for binary
CSS matrices.  For a general APM lift, a complete search uses all \(LP\)
qubit positions as roots; a verified automorphism permits reduction to orbit
representatives.  The all-one vector still lies in each check row space, so all
vectors in the opposite kernel have even weight.  Thus exclusion through
weight \(d_0-2\) proves \(d\ge d_0\) for even \(d_0\).  Applying cyclic root
reduction or a polynomial permanent bound to an APM lift requires separate
verification of the corresponding symmetry or algebraic hypotheses.

The fixed \(\qcode{576}{294}{12}\) APM instance in the project catalogue
\cite{okadaKasaiData2026} gives a complementary example beyond termwise
commutativity: all 16 block-row pairs split into eight equal-relative-map pairs,
and complete exclusion through weight 10, together with even weight, proves
\(d\ge12\).

The standard parent construction with consecutive selected rows has the
restriction \(L\ge4J\), and hence full-rank design rate
\(R_{\rm des}=1-2J/L\ge1/2\).  Removing this restriction permits searches in
the range \(2J<L<4J\), including design rates below one half and smaller
values of \(L\) or \(P\).  Candidate codes in this range can be evaluated for
GALA- or AOD-oriented designs \cite{yangDuckeringDua2026GALA} and analogous
neutral-atom constraints \cite{zhao2026UltraHighRate} through their
symmetries, rank, distance, short cycles, and implementation schedules.

\section{Conclusion}

We developed a pair-partition construction in which paired relative maps give
CSS orthogonality, pairing graphs give label-independent girth restrictions,
and completed lifts are checked directly.  CPM labels reduce these
constraints to linear paired-difference equations.  The later APM extension
shows that the cancellation principle is more general than cyclic shifts,
and its structured specialization recovers the selected check matrices
called active in the Kasai construction.  The resulting list contains
thirty-six exact-distance CPM codes.  These include the high-rate
\(\qcode{1414}{812}{12}\) code and the girth-eight
\(\qcode{472}{122}{16}\), \(\qcode{584}{150}{18}\),
\(\qcode{1112}{282}{20}\), \(\qcode{1336}{338}{22}\),
\(\qcode{1630}{656}{20}\), \(\qcode{1784}{450}{24}\), and
\(\qcode{2230}{896}{24}\) codes.
The permanent construction gives the complementary universal bound
\(d\le(J+1)!\) under the stated rank and size conditions.

The pair-partition array determines CSS orthogonality and
exponent-independent cycle obstructions.  Binary rank and Tanner girth are
computed on the completed lift; quantum-distance lower bounds are certified
by the exhaustive search of Theorem~\ref{thm:distance-search-correctness} and
upper bounds by explicit non-stabilizer zero-syndrome vectors.  Exact distance
is reported when these bounds agree.

The present results concern reproducible finite codes.  The universal
girth-12 upper bound applies to the complete-protograph CPM family;
asymptotic distance growth for fixed \(J,L\) remains an open direction.

\vspace{-0.5\baselineskip}
\section*{Reproducibility Materials}

The project webpage \cite{okadaKasaiData2026} provides the code parameters,
fixed exponent or affine-label arrays, CPM and APM matrices, displayed girth
cycles, contributed code files, distance-search records, explicit logical
vectors, and the files used for the reported claims.  The GitHub
repository \cite{okadaKasaiCode2026} provides the software for construction,
exponent search, girth verification, and distance verification.

\vspace{-0.5\baselineskip}
\section*{Acknowledgment}

The motivation for the code construction developed in this work arose from
discussions with Hengyun Zhou, whom we thank.  We also thank Nishad Maskara for providing the
three check-weight-eight codes
\(\qcode{232}{62}{12}\), \(\qcode{248}{66}{12}\), and
\(\qcode{296}{78}{12}\), and the weight-20 logical vectors that establish
the matching upper bound for \(\qcode{944}{478}{20}\).  We also thank
Jong Yeon Lee for providing 19 additional CPM-PP instances.  Selected
instances are included in Table~\ref{tab:finite-summary} and discussed in
Section~\ref{sec:numerical-results}; data for the full set are available on
the project webpage~\cite{okadaKasaiData2026}.

\vspace{-0.5\baselineskip}
\enlargethispage{2\baselineskip}
\bibliographystyle{IEEEtran}
\bibliography{references}

\end{document}